\documentclass[11pt]{article}
\usepackage{url}
\usepackage{graphicx}
\usepackage{mathtools}
\usepackage[mathscr]{eucal}
\font\greek=grmn1000 at 11pt
\usepackage{theorem}
{\theorembodyfont{\rmfamily}
\newtheorem{definition}{\bfseries Definition}[section]}
\newtheorem{theorem}{Theorem}[section]
\newtheorem{proposition}{Proposition}[section]
\newsavebox{\proofbox}
\savebox{\proofbox}{%
  \begin{picture}(7,7)\put(0,0){\framebox(7,7){}}\end{picture}%
}
\newenvironment{proof}[1][]{%
  \list{}{\leftmargin0pt
    \rightmargin\leftmargin}%
  \item[]{\hspace*{1em}\textbf{Proof}\ #1}%
}
{\hspace*{\fill}{\usebox{\proofbox}}\endlist}
\begin{document}
\title{Fuzzy Aristotelian Diagrams}
\author{Apostolos Syropoulos\\ Xanthi, Greece\\ \texttt{asyropoulos@yahoo.com}}
\maketitle
\begin{abstract}
After a concise introduction to the square of opposition, in particular, and, Aristotelian Diagrams, in general, I describe
how one can create a mathematical universe to host these objects. Since these objects assume that the underlying logic is the
bivalent logic, I have used these objects as a starting point to introduce fuzzy Aristotelian Diagrams and  
describe a mathematical formulation of them. In addition, I outline the cdharacteristrics of a mathematical universe 
that hosts them.
\end{abstract}
\section{Introduction}
The square of opposition is a diagram that graphically represents a collection of logical relationships~\cite{sep-square}. 
This diagram was particularly useful for logicians. More specifically, they used it in their work for more than 2000 years. The ideas 
behind the square of opposition have their roots in the work of the Greek philosopher Aristotle, who lived in the fourth century BC. 
De Interpretatione 6–9 discusses the rule of contradictory pairs, which forms the basis of the square of opposition. This rule states 
that for any contradictory pair of statements, only one is true and the other is false. {\em De Interpretatione} or 
{\em On Interpretation} (Greek: {\greek Per'i <ermhne'ias}; Latinized Greek: Peri Hermeneias) is the second book of 
Aristotle’s Organon and deals with the relationship between language and logic~\cite{jones2010}. As it happens with many 
ideas and scientific results of the past, the square of opposition has been severely criticized in recent decades; 
nevertheless, it is still something that is frequently used by logicians.

The underlying logic of the square of opposition is the classical bivalent logic. Some call it Aristotelian logic, although 
in chapter 9 of De Interpretatione, Aristotle asked whether it makes sense to say that a sentence about a future event 
that can occur or cannot occur is true or false. Thus, the so-called problem of future contingents indirectly speculates 
about more than two truth values.

We use fuzzy sets and their extensions, such as Krassimir Atanassov's {\em intuitionistic} fuzzy sets~\cite{atanassov2012}, 
as mathematical tools to address vagueness. But what is vagueness? Generally speaking, David Lanius~\cite{lanius2021} correctly 
points out that vagueness leads to the Sorites paradox, borderline cases, and the violation of the logical principle of bivalent 
logic; readers who are not familiar with the Sorites paradox and borderline cases should consult any presentation in layman’s 
terms (e.g., see~\cite{syropoulos2023b}). However, Lanius claims that vagueness is inherently negative. In my own opinion, this 
claim is wrong since vagueness is a fundamental property of our world. Currently, quantum mechanics is our best theory for 
accurately describing the building blocks of our universe, and it is considered a vague theory! Why? Because the properties of 
our cosmos' building blocks are influenced by the underlying property of vagueness (refer 
to~\cite{syropoulos2021, syropoulos2021b, syropoulos2023}). Of course, vagueness is a semantic property of linguistic 
expressions~\cite{bones2020}, but this is something I will not further explore.

Probability theory is a branch of mathematics that studies the outcome of random experiments (e.g., the roll of a die, 
the results of horse racing, the maximum and the minimum temperature in New York in two weeks from today, etc.). For any random experiment, 
we must be aware of all possible outcomes and determine whether the result of the experiment is truly random (i.e., unpredictable). 
Depending on the characteristics of the experiment, {\em probabilities} can be assigned to the possible outcomes. These probabilities are 
numbers that belong to the unit interval. For example, when we throw a balanced, six-faced die, the possible outcomes are 1, 2, 3, 4, 5, 
and 6, and to each of them we assign a probability of 1/6. Thus, the probability that the next roll of the die will be 6 is 1/6. Some 
authors consider that probabilities and membership degrees are the same thing or at least two facets of the same concept. In fact, 
Zadeh, the founder of fuzzy set theory, believed this. However, this idea is totally wrong from a pragmatic and, more generally, a 
philosophic perspective. It's important to note that probabilities and membership degrees measure entirely distinct entities. A 
probability is the likelihood degree of some event, while a membership degree is a truth value (i.e., something that {\em asserts} to what 
degree something has a particular property or to what degree something will happen). This is precisely why I do not find the idea 
of proposing a probabilistic version of the square of opposition~\cite{pfeifer2017} to be useful.

It is my firm belief that a reinterpretation of the square of opposition using fuzzy mathematics is quite reasonable. For example, 
one such work uses functional degrees of inclusion to reinterpret contradiction~\cite{madrid2024}. However, I am not merely interested 
in fuzzifying the square of opposition but in making it ``alive'' in a vague universe. Such a universe could be a vague 
category~\cite{mclane1998}. Leander Vignero~\cite{vingero2021} has already examined some properties of a category 
whose objects are squares of opposition. In addition, Alexander De~Klerck, Leander Vignero, and Lorenz Demey~\cite{deKlerck2023} 
defined different morphisms between squares of opposition and thus defined different categories.

\paragraph{Plan of the paper} First I will properly introduce the square of opposition. Next, I will explain how to fuzzify the square of 
opposition using intuitionistic fuzzy logic, and afterward, I will demonstrate how to construct a suitable universe for these concepts. 
The paper concludes with the customary conclusions section.
\section{Traditional Aristotelian Diagrams}\label{crisp:diagrams}
The square of opposition looks surprisingly  like a commutative diagram (see Figure~\ref{trad::sq:opp}); however, its corners are 
{\em logical forms} and not some ordinary objects. The table that follows describes these four logical forms.
\begin{center}
\begin{tabular}{l|l}\hline
FORM                  & TITLE\\\hline
Every $S$ is $P$      & Universal Affirmative\\
No $S$ is $P$         & Universal Negative\\
Some $S$ is $P$       & Particular Affirmative\\
Some $S$ is not $P$   &	Particular Negative\\ \hline
\end{tabular}
\end{center}
\begin{figure}[t]
\begin{center}
\includegraphics[scale=1]{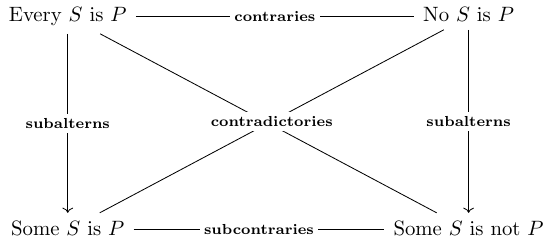}
\end{center}
\caption{The traditional square of opposition.}\label{trad::sq:opp}
\end{figure}
If the diagram in Figure~\ref{trad::sq:opp} is not a commutative diagram, then what is the meaning of this diagram? Put differently, what 
facts does this diagram convey? The following {\em theses}, as they are called by the author of~\cite{sep-square}, are expressed in 
the diagram of Figure~\ref{trad::sq:opp}.
\begin{itemize}
\item ``Every $S$ is $P$'' and ``Some $S$ is not $P$'' are contradictories.
\item ``No $S$ is $P$'' and ``Some $S$ is $P$'' are contradictories.
\item ``Every $S$ is $P$ and ``No $S$ is $P$ are contraries.
\item ``Some $S$ is $P$'' and ``Some $S$ is not $P$ are subcontraries.
\item ``Some $S$ is $P$'' is a subaltern of ``Every $S$ is $P$.''
\item ``Some $S$ is not $P$’ is a subaltern of ``No $S$ is $P$.''
\end{itemize} 
Figure~\ref{mod::sq:opp} shows a {\em modern} version of the square of opposition. People rarely use this diagram because 
it conveys little information.
\begin{figure}[t]
\begin{center}
\includegraphics[scale=1]{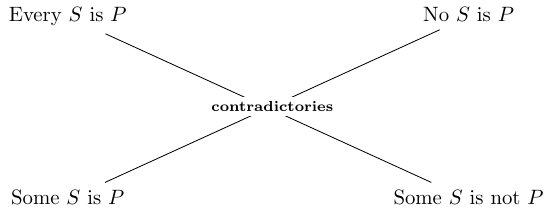}
\end{center}
\caption{A modern version of the square of opposition.}\label{mod::sq:opp}
\end{figure}

Mathematically speaking, a square of opposition is an {\em Aristotelian} diagram and {\em Boolean} algebras describe these diagrams. 
But what is a Boolean algebra? The following definition explains what a Boolean algebra is.\footnote{Weisstein, Eric W. ``Boolean Algebra.'' 
From MathWorld--A Wolfram Web Resource. 
\url{https://mathworld.wolfram.com/BooleanAlgebra.html}} 
\begin{definition}
A Boolean algebra is a set $B$ of elements $a, b,\ldots$ with the following properties:
\begin{enumerate}
 \item  $B$ is equipped with two binary operations, $\wedge$ and $\vee$, which satisfy the {\em idempotent} laws 
 \begin{displaymath}
 a\wedge a=a\vee a=a,
 \end{displaymath}
 the {\em commutative} laws
 \begin{align*}
a \wedge b	&=	b \wedge a\\	
a \vee b	&=	b \vee a,
 \end{align*}
 and the {\em associative} laws
 \begin{align*}
a \wedge (b \wedge c) &= (a\wedge b)\wedge c\\	
a\vee (b \vee c)      &= (a\vee b) \vee c.
 \end{align*}
\item In addition, these operations satisfy the {\em absorption} law
 \begin{displaymath}
 a\wedge (a \vee b)=a \vee (a \wedge b)=a.
 \end{displaymath} 
\item Also, the operations are {\em mutually distributive}
 \begin{align*}
 a\wedge (b \vee c) &= (a\wedge b)\vee (a\wedge c)\\	
 a\vee (b\wedge c)  &= (a\vee b)\wedge (a \vee c).
 \end{align*}
\item $B$ contains universal bounds $0$ and $1$  which satisfy
 \begin{align*}
 0\wedge a &=0 \\	
 0\vee a	  &= a\\	
 1\wedge a &= a\\	
 1\vee a	&=1.
 \end{align*}
\item $B$ has a unary operation $\neg a$ of {\em complementation}, which obeys the following laws:
 \begin{displaymath}
 a\wedge \neg a = 0\quad\text{and}\quad a\vee \neg a=1.
 \end{displaymath}
\end{enumerate}
\end{definition}
Using this defintion, we define Aristotelian diagrams~\cite{deKlerck2023}:
\begin{definition}
An {\em Aristotelian} diagram $D$ is a pair $(F, B)$, where $B$ is the underlying set of the Boolean algebra 
$(B, \wedge_{B}, \vee_{B} , \neg_{B} , 1_B , 0_B)$ and $F\subseteq B$. 
When the Boolean algebra is clear from context, we usually omit the subscripts from $\wedge$, $\vee$, etc.
\end{definition}
The next thing we need to know is how the theses described above are expressed in the language of Boolean algebras.
\begin{definition}\label{arist:diag:rel}
Assume that $(F,B)$ is an Aristolenian diagram. Then, we say that $x, y\in B$ are in:
\begin{enumerate}
\item $B$-bi-implication ($\text{BI}_B$) if and only if $x = y$;
\item $B$-left-implication ($\text{LI}_B$) if and only if $x < y$;
\item $B$-right-implication ($\text{RI}_B$) if and only if $y < x$;
\item $B$-contradictory ($\text{CD}_B$) if and only if $x \mathbin{\wedge_B} y = 0_B$ and $x \mathbin{\vee_B} y = 1_B$, 
      that is, $x = \mathop{\neg_{B}} y$;
\item $B$-contrary ($\text{C}_B$) if and only if $x \mathbin{\wedge_B} y = 0_B$ and $x\mathbin{\vee_B} y\neq 1_B$, 
      that is, $x \mathrel{<_B} \mathop{\neg_{B}} y$;
\item $B$-subcontrary ($\text{SC}_B$) if and only if $x\mathbin{\wedge_B} y\neq 0_B$ and $x\mathbin{\vee_B} y = 1_B$, that is,
      that is, $x \mathrel{>_B} \mathop{\neg_{B}} y$;
\item $B$-unconnectedness ($\text{Un}_B$) if and only if none of the above holds.
\end{enumerate}
\end{definition}
The above relations are called {\em logical} relations and are denoted by $\Re$. Also, the first three relations are  called 
{\em implication} relations and relations 4–6 are called {\em opposition} relations. The following definition provides 
a formal mechanism for relating Aristolenian diagrams.
\begin{definition}
An Aristotelian isomorphism $f: (F_1, B_1)\rightarrow(F_2,B_2)$ is a bijection $f: F_1\rightarrow F_2$
such that for all logical relations $R\in\Re$  and all $x, y\in F_1$ we have that 
$x\mathrel{R_{B_1}} y$ if and only if $f(x)\mathrel{R_{B_2}}f(y)$.
\end{definition}
This is strong way to go from one diagram to another and is not bery useful. A weaker one make use of the following order relation:
\begin{definition}
There is an i{\em nformativity order} $\le_i$ on $\Re$ which is given by: $\text{Un}\mathrel{\le_i}\text{LI}$, 
$\text{Un}\mathrel{\le_i}\text{RI}$, $\text{Un}\mathrel{\le_i}\text{C}$, 
$\text{Un}\mathrel{\le_i}\text{SC}$, $\text{LI}\mathrel{\le_i}\text{BI}$, 
$\text{RI}\mathrel{\le_i}\text{BI}$, $\text{C}\mathrel{\le_i}\text{CD}$, and $\text{SC}\mathrel{\le_i}\text{CD}$.
\end{definition}
{\em Infomorphisms} are defined as follows:
\begin{definition}
An infomorphism $f:(F_1,B_1)\rightarrow (F_2, B_2)$ is a function that satisfies the following condition: 
for all $x, y\in F_1$ it holds that if $x\mathrel{R_{B_1}}y$, then $f(x)\mathrel{S_{B_2}}f(y)$ with $R\mathrel{\le_i}S$.
\end{definition}
\section{Contradiction and Vagueness}
If we want to interpret diagram~\ref{trad::sq:opp} using fuzzy logic (i.e., a logic of vagueness), then we must understand
how contradiction is represented in such a logic. Trillas,  Alsina, and Jacas~\cite{trillas1999} started from the classical
law of contradiction, which says that for all propositions $p$, it is impossible for both $p$ and not $p$ to be true,
that is, $p\Rightarrow\neg p$. Such propositions are called self-contradictory. In addition, if and only if $p\Rightarrow\neg q$
is true, then $p$ and $q$ are contradictory. However, I have to remark that the standard law of contradiction states that for all
$p$, $p\wedge\neg p$ is always false.  Trillas et al., assume that in general $p\Rightarrow q\equiv\neg p\vee(p\wedge q)$, that is,
the implication operator is the so-called material  implication. If we assume that $p$ and $q$ are fuzzy sets (i.e., the
equivalent of predicates in classical logic), then we can define the {\em degree of contradiction} as follows:
\begin{definition}
Suppose that $A, B\in [0,1]^{E}$, and $J$ is an implication operator, and $N$ a negation operator, then $A$ is 
contradictory to $B$ with degree equal to $r\in[0,1]$, when $J(A(x),N(B(x))=r$. 
\end{definition}
From here we can define distances and other things (e.g., see~\cite{cubillo2005, torres2007}).

Assume that we assign truth values to vague propositions that belong to the unit interval; nevertheless, we do not follow 
the principles of fuzzy mathematics. Indeed, Nicholas Smith~\cite{smith2008} proposed that contradictory propositions
 should always have a truth value equal to $0.5$. In my opinion, it is more natural to accept Smith’s proposal, and here is why. 
Assume that an eagle flies in the sky and consider the statement, ``The eagle is in a cloud.'' According to classical logic, 
this statement can be either true or false. However, when the eagle is partly inside a cloud and, consequently, partly outside 
the cloud, then, strictly speaking, the statement is false, but it is closer to reality to assume that the statement and its 
negation are both true. After all, this is why paraconsistent logics (i.e., logics that assume that some contradictions are 
true) are considered bivalent models of vagueness. In conclusion, it makes sense to assume that the degree to which a vague 
proposition is a contradiction is always $0.5$. This means that the degree to which it is not a contradiction should be $0.5$. 
I think it makes sense that this degree is less than or equal to $0.5$. In addition, this value should not be derivable from 
some formula. In different words, I propose that ``intuitionistic'' fuzzy logic (IFL) is the ideal tool to model vague 
contradictions. Figure~\ref{if::trad::sq:opp} shows how the traditional square of opposition can be modified when the 
underlying logic is IFL.
\begin{figure}[t]
\begin{center}
\includegraphics[scale=1]{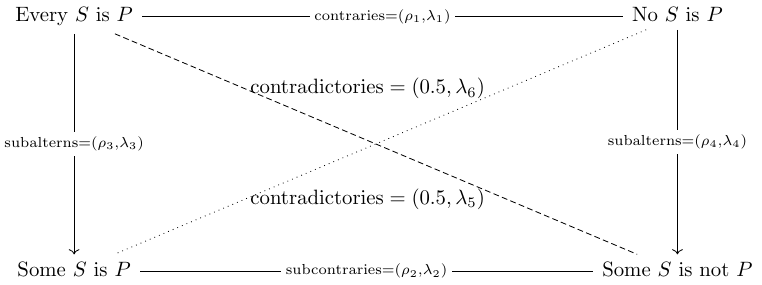}
\end{center}
\caption{An ``intuitionistic'' fuzzy version of the traditional square of opposition.}\label{if::trad::sq:opp}
\end{figure}
\section{``Intuitionistic'' Fuzzy Aristotelian Diagrams}
``Intuitionistic'' fuzzy Aristotelian diagrams are subsets of ``intuitionistic'' fuzzy Boolean algebras. To properly define 
these structures, we need some auxiliary definitions stated in~\cite{tripathy2013}.
\begin{definition}
 An {\em ``intuitionistic'' fuzzy relation} is an ``intuitionistic'' fuzzy subset of $X\times Y$; that is, the set 
$R$ given by
\begin{displaymath}
R = \bigl\{\langle (x, y), \mu_{R}(x,y), \nu_{R}(x,y)\rangle\mathrel{\bigm|} x\in X, y\in Y\bigr\},
\end{displaymath}
where $\mu_{R},\nu_{R}: X\times Y\rightarrow[0,1]$, satisfy the condition $0\le\mu_{R}(x,y) + \nu_{R}(x,y)\le 1$, 
for every $(x, y)\in X \times Y$.
\end{definition}
\begin{definition}
An``intuitionistic'' fuzzy relation in $X\times X$ is
\begin{description}
\item[reflexive] if for every $x\in X$, $\mu_{R}(x,x) = 1$ and $\nu_{R}(x,x) = 0$.
\item[perfectly antisymmetric] if for every $(x,y)\in X\times X$ with $x\neq y$ and
$\mu_R(x,y) > 0$ or $\mu_{R}(x,y) = 0$ and at the same time $\nu_{R}(x,y) < 1$, then
$\mu_{R}(x,y) = 0$ and $\nu_{R}(x,y) = 1$.
\item[transitive] if $R\circ R\subseteq R$, where $\circ$ is $\max$-$\min$ and $\min$-$\max$ composition, that is
\begin{align*}
\mu_{R}(x,z)\ge\max_{y}\Bigl[\min\bigl\{\mu_{R}(x,y),\mu_{R}(x,z)\bigr\}\Bigr]\\ 
\intertext{and}
\nu_{R}(x,z)\le\min_{y}\Bigl[\max\bigl\{\nu_{R}(x,y),\nu_{R}(y,z)\bigr\}\Bigr].
\end{align*}
\end{description}
\end{definition}
\begin{definition}
An ``intuitionistic'' fuzzy relation $R$ on $X\times X$ is said to be an ``intuitionistic'' fuzzy {\em partially 
ordered} relation if $R$ is reflexive, perfectly antisymmetric, and transitive.
\end{definition}
\begin{definition}
A crisp set $X$ on which an ``intuitionistic'' fuzzy partial ordering $R$ is defined is said to be an
``intuitionistic'' fuzzy {\em lattice} if and only if for any two element set $\{x, y\}\subset X$, the 
least upper bound (lub) and greatest lower bound (glb) exist in $X$. We denote the lub of $\{x, y\}$ by 
$x\vee y$ and the glb of $\{x, y\}$ by $x\wedge y$.
\end{definition}
\begin{definition}
An ``intuitionistic'' fuzzy lattice $(X,R)$ is {\em distributive} if and only if for all $a, b, c\in X$,
\begin{displaymath}
a\wedge(b\vee c ) = (a\wedge b)\vee (a\wedge c)\quad\text{and}\quad
a\vee(b\wedge c) = (a\vee b)\wedge (a\vee c).
\end{displaymath}
\end{definition}
\begin{theorem}
If $(X, R)$ is a {\bfseries complemented} (i.e., it has a unary operation $\neg a$ of complementation),
distributive ``intuitionistic'' fuzzy lattice, then the two De Morgan’s Laws
\begin{displaymath}
\neg(a\vee b) = \neg a\wedge\neg b\quad\text{and}\quad \neg(a\wedge b) = \neg a \vee\neg b, 
\end{displaymath}
hold for all $a, b\in X$.
\end{theorem}
Obviously, we can use the standard complementaion of ``intuitionistic'' fuzzy sets or a variation of it
suitable for our case.
\begin{definition}
A complemented distributive `intuitionistic'' fuzzy lattice is an `intuitionistic'' fuzzy Boolean algebra.
\end{definition}
And now comes the definition of fuzzy Aristotelian diagrams:
\begin{definition}
Assume that $B=(X, R)$ is an ``intuitionistic'' fuzzy Boolean algebra. Then, a fuzzy Aristotelian diagram $D$ is a 
pair $(F, B)$, where $F\subseteq X$.
\end{definition}
\section{A Category of Fuzzy Aristotelian Diagrams}
Let's now focus on how we can define a universe of fuzzy Aristotelian diagrams. First, it is necesssary to define 
morphisms between ``intuitionistic'' fuzzy Boolean algebra. In the classical case, a Boolean homomorphism is a 
function between two Boolean algebras that preserves their algebraic structure. Since, an ``intuitionistic'' fuzzy 
Boolean homomorphism is an extension of the classical concept, it is expected that the new definition will be an
extension of the classical one. Indeed, here is how we can define these new morphisms: 
\begin{definition}
A function $f$ between two ``intuitionistic'' fuzzy Boolean algebras $B_{1}$ and $B_{2}$ is a {\em homomorphism} 
if it preserves the operations: 
\begin{enumerate}
\item $f(x \mathrel{\vee_{B_1}} y) = f(x) \mathrel{\vee_{B_2}} f(y)$,
\item $f(\mathop{\neg_{B_1}}x) = \mathop{\neg_{B_2}}f(x)$, 
\item $f(0_{B_1})=0_{B_2}$, and 
\item $f(1_{B_1})=1_{B_2}$.
\end{enumerate} 
As in the case of (crisp) Aristotelian diagrams, this definition is too strong and not particularly useful. Thus,
I propose the following definition:
\end{definition}
\begin{definition}
A {\em fuzzy infomorphism} $f:(F_1,B_1)\rightarrow (F_2, B_2)$ is a function from $F_1$ to $F_2$ that satisfies the following condition: 
for all $x, y\in F_1$ it holds that if $x\mathrel{R_{B_1}}y$, then $f(x)\mathrel{S_{B_2}}f(y)$ with $R\mathrel{\le_i}S$.
\end{definition}
The following statement can be easily proved.
\begin{proposition}
The collection \textbf{IFBA} of all ``intuitionistic'' fuzzy Boolean algebras and their fuzzy infomorphism is a category.
\end{proposition}
\begin{proof}
The objects of the category are ``intuitionistic'' fuzzy Boolean algebras and the morphisms are the fuzzy infomorphism.
The min-max composition is associative (see~\cite[p. 81]{syropoulos2020}) and the max-min composition is associative for three
fuzzy relations (see~\cite{Shakhatreh2020}). Therefore, composition of fuzzy infomorphism is associative. The identity morphism
is one that leaves an  ``intuitionistic'' fuzzy Boolean algebra as is. Thus, \textbf{IFBA} is a category.
\end{proof}

The various relations presented in the definition~\ref{arist:diag:rel} can be easily extended in the our case. For example,
bi-implication is meaningful to a degree. In particular, given an ``intuitionistic'' fuzzy Boolean algebra $(F,B)$, then 
for $x,y\in F$, we can say that $x=y$ if and only if the membership degree of $x\mathrel{R_{B}}y$ is equal to the nonmembership 
degree of $x\mathrel{R_{B}}y$. However, since this is too strict, we can say that the difference of the membership and the nonmembership 
degrees are within a specific range. For example, we can say that the differences can be less that $0.01$.
\section{Conclusions}
I have described (crisp) Aristotelian diagrams and I introduced a fuzzy version of these diagrams. In addition,
I described the general characteristics of a mathematical universe that can host them. The next thing to do is
to examine this mathematical universe and reveal its properties.

\end{document}